\documentclass[a4paper]{article}
\pdfoutput=1
\usepackage[final]{microtype}
\usepackage{local-macros}
\usepackage{local-amsthm}
\usepackage{newpxtext}
\usepackage{newpxmath}
\usepackage[numbers]{natbib}

\title{What should a generic object be?}
\author{Jonathan Sterling}
\date{Updated \today}

\begin{document}
\maketitle

\begin{abstract}
  Jacobs has proposed definitions for (weak, strong, split) generic objects for
  a fibered category; building on his definition of (split) generic objects,
  Jacobs develops a menagerie of important fibrational structures with
  applications to categorical logic and computer science, including
  \emph{higher order fibrations}, \emph{polymorphic fibrations},
  \emph{$\lambda2$-fibrations}, \emph{triposes}, and others.  We observe that a
  {split generic object} need not in particular be a {generic object} under the
  given definitions, and that the definitions of polymorphic fibrations,
  triposes, \etc are strict enough to rule out some fundamental examples: for
  instance, the fibered preorder induced by a partial combinatory algebra in
  realizability is not a tripos in this sense.
  We propose a new alignment of terminology that emphasizes the forms of
  generic object appearing most commonly in nature, \ie in the study of
  internal categories, triposes, and the denotational semantics of
  polymorphism. In addition, we propose a new class of \ul{acyclic} generic
  objects inspired by recent developments in higher category theory and the
  semantics of homotopy type theory, generalizing the \emph{realignment}
  property of universes to the setting of an arbitrary fibration.
\end{abstract}

\tableofcontents

\NewDocumentCommand\FAM{m}{\Kwd{Fam}\prn{#1}}

\section{Introduction}

Since the latter half of the 20th century, \emph{fibered category theory} or
the \emph{theory of fibrations} has played an important background role in both
the applications and foundations of category theory~\citep{sga:1,benabou:1985}.
Fibered categories, also known as fibrations, are a formalism for manipulating
categories that are defined \emph{relative} to another category, generalizing
the way that ordinary categories can be thought of as being defined relative to
the category of sets.

The sense in which ordinary category theory is pinned to the category of sets
can be illustrated by considering the definition of when a category $\CCat$
``has products'':

\begin{definition}\label[definition]{def:cat-has-products}
  A category $\CCat$ has \textbf{products} when for any indexed family
  $\brc{E_i\in \CCat}\Sub{i\in I}$ of objects, there exists an object
  $\Prod{i\in I}{E_i}\in\CCat$ together with a family of morphisms
  $\Mor[p_k]{\Prod{i\in I}E_i}{E_k}$ such that for any family of morphisms
  $\Mor[h_k]{H}{E_k}$ there exists a unique morphism $\Mor[h]{H}{\Prod{i\in
  I}E_i}$ factoring each $h_k$ through $p_k$.
\end{definition}

In the above, the dependency on the category $\SET$ is clear: the indexing
object $I$ is a set. If we had required $I$ to be drawn from a proper
subcategory of $\SET$ (\eg finite sets) or a proper supercategory (\eg
classes), the notion of product defined thereby would have been different. The purpose
of the formalism of \emph{fibered categories} is to explicitly control the
ambient category that parameterizes all indexed notions, such as products,
sums, limits, colimits, \etc.

\begin{remark}[Relevance to computer science]\label{rem:small-complete-category}
  The ability to explicitly control the parameterization of products and sums
  is very important in theoretical computer science, especially for the
  denotational semantics of \emph{polymorphic types} of the form
  $\forall\alpha. \tau\brk{\alpha}$. Such a polymorphic type should be
  understood as the product of all $\tau\brk{\alpha}$ indexed in the ``set'' of
  all types, but a famous result of \citet{freyd:1964} shows that if a category
  $\CCat$ has products of this form \emph{parameterized in $\SET$}, then
  $\CCat$ must be a preorder. Far from bringing to an early end the study of
  polymorphic types in computer science, awareness of Freyd's result sparked
  and guided the search for ambient categories other than $\SET$ in which to
  parameterize these
  products~\citep{pitts:1987,hyland:1988,hyland-robinson-rosolini:1990}.
  Fibered category theory provides the optimal language to understand all such
  indexing scenarios, and the textbook of \citet{jacobs:1999}, discussed at
  length in the present paper, provides a detailed introduction to the
  applications of fibered category theory to theoretical computer science.
\end{remark}

\subsection{Introduction to fibered categories}

Before giving a general definition, we will see the way that fibered
categorical language indeed makes parameterization explicit by considering the
\emph{prototype} of all fibered categories, the category $\FAM{\CCat}$ of
$\SET$-indexed families of objects of a category $\CCat$.

\begin{construction}[The category of families]\label{con:family-fibration}
  We define $\FAM{\CCat}$ to be the category of $\SET$-indexed families in $\CCat$, such that
  \begin{enumerate}
    \item an object of $\FAM{\CCat}$ is a family $\brc{E_i\in\CCat}\Sub{i\in I}$ where $I$ is a set,
    \item a morphism $\Mor{\brc{E_i}\Sub{i \in I}}{\brc{F_j}\Sub{j\in J}}$ in $\FAM{\CCat}$ is given by a function $\Mor[u]{I}{J}$ together with for each $i\in I$ a morphism $\Mor[\bar{u}_i]{E_i}{F\Sub{ui}}$.
  \end{enumerate}

  There is an evident functor $\Mor[p]{\FAM{\CCat}}{\SET}$ taking $\prn{\brc{E_i}\Sub{i\in I}}$ to $I$.
\end{construction}

\begin{construction}[Fiber categories]
  For each $I\in \SET$, we may define the \emph{fiber} $\FAM{\CCat}_I$ of
  $\FAM{\CCat}$ over $I$ to be the category of $I$-indexed families $\brc{E_i\in
  \CCat}\Sub{i\in I}$ in $\CCat$, with morphisms $\Mor{\brc{F_i}\Sub{i\in
  I}}{\brc{E_i}\Sub{i\in I}}$ given by morphisms $\Mor[h_i]{F_i}{E_i}$ for each $i\in
  I$.
\end{construction}

More abstractly, the fiber category $\FAM{\CCat}_I$ is the following pullback:
\[
  \DiagramSquare{
    ne = \FAM{\CCat},
    se = \SET,
    east = p,
    south = I,
    sw = \Kwd{1},
    nw = \FAM{\CCat}_I,
    width = 2.5cm,
    nw/style = pullback,
  }
\]

\begin{construction}[Reindexing functors]\label{con:reindexing}
  For any function $\Mor[u]{J}{I}$, there is a corresponding reindexing functor
  $\Mor[u^*]{\FAM{\CCat}_I}{\FAM{\CCat}_J}$ that restricts an $I$-indexed family
  into a $J$-indexed family by precomposition.
\end{construction}

With the reindexing functors in
hand, can now rephrase the condition that $\CCat$ has products
(\cref{def:cat-has-products}) in terms of $\FAM{\CCat}$.

\begin{proposition}
  A category $\CCat$ has products if and only if for each product projection
  function $\Mor[\pi\Sub{I,J}]{I\times J}{I}$, the reindexing functor
  $\Mor[\pi\Sub{I,J}]{\FAM{\CCat}\Sub{I}}{\FAM{\CCat}\Sub{I\times J}}$ has a
  right adjoint $\Mor[\Prod{\prn{I,J}}]{\FAM{\CCat}\Sub{I\times
  J}}{\FAM{\CCat}_I}$ such that the following \emph{Beck--Chevalley} condition
  holds: for any function $\Mor[u]{K}{I}$, the canonical natural transformation
  $\Mor{u^*\circ \Prod{\prn{I,J}}}{\Prod{\prn{K,J}}\circ \prn{u\times
  \Idn{J}}^*}$ is an isomorphism.
\end{proposition}

The characterization of products in terms of the category of families may seem
more complicated, but it has a remarkable advantage: we can replace
$\Mor{\FAM{\CCat}}{\SET}$ with a different functor satisfying similar
properties in order to speak more generally of when one category has
``products'' that are parameterized in another category. The properties that
this functor has to satisfy for the notion to make sense are embodied in the
definition of a \emph{fibration} or \emph{fibered category}; a functor
$\Mor{\ECat}{\BCat}$ will be called a fibration when it behaves similarly to
the functor projecting the parameterizing object from a category of families of
objects. We begin with an auxiliary definition of \emph{cartesian morphism}:

\begin{definition}
  Let $\Mor[p]{\ECat}{\BCat}$ and let $\Mor{E}{F}$ be a morphism in $\ECat$, which we depict as follows:
  \[
    \DiagramSquare{
      nw = E,
      ne = F,
      sw = pE,
      se = pF,
      west/style = {|->},
      east/style = {|->},
      height = 1.5cm,
    }
  \]

  In the diagram above, we say that $\Mor{E}{F}$ \textbf{lies over}
  $\Mor{pE}{pF}$.
  We say that $\Mor{E}{F}$ is \textbf{cartesian} in $p$ when for any morphism
  $\Mor{H}{F}$ in $\ECat$ and $\Mor{pH}{pE}$ such that the former lies over the composite $pH\to pE\to pF$ in $\BCat$, there exists a unique
  morphism $\Mor{H}{E}$ lying over $\Mor{pH}{pE}$ such that $\Mor{H}{F}$ lies
  over the composite $\Mor{pH}{pE}$ as depicted below:
  \[
    \begin{tikzpicture}[diagram]
      \SpliceDiagramSquare<l/>{
        nw = H,
        sw = pH,
        se = pE,
        ne = E,
        west/style = {|->},
        east/style = {|->},
        height = 1.5cm,
        north/style = {exists,->},
        north/node/style = upright desc,
        north = \exists!,
      }
      \SpliceDiagramSquare<r/>{
        glue = west, glue target = l/,
        east/style = {|->},
        height = 1.5cm,
        ne = F,
        se = pF,
      }
      \draw[->,bend left=30] (l/nw) to (r/ne);
    \end{tikzpicture}
  \]
\end{definition}

\begin{remark}[Explication of cartesian maps]\label[remark]{rem:cartesian-maps}
  Returning to our example of the category of families $\FAM{\CCat}$ over
  $\SET$, we can make sense of the notion of a cartesian map. Given a function
  $\Mor[u]{J}{I}$ of indexing sets, the reindexing functor $u^*$ takes
  $I$-indexed families to $J$-indexed families. Given an $I$-indexed family
  $\brc{E_i}\Sub{i\in I}$, we may define a morphism $\Mor{u^*\brc{E_i}\Sub{i\in
  I}}{\brc{E_i}\Sub{i\in I}}$ in $\CCat$ whose first component is
  $\Mor[u]{J}{I}$ and whose second component is the identity function
  $\Mor{E\Sub{uj}}{E\Sub{uj}}$ at each $j\in J$.  The morphism
  $\Mor{u^*\brc{E_i}\Sub{i\in I}}{\brc{E_i}\Sub{i\in I}}$ is then
  \emph{cartesian} in $\Mor[p]{\FAM{\CCat}}{\SET}$.
\end{remark}

\begin{exercise}\label{exc:cart-lift-in-fam}
  Verify that the morphism $\Mor{u^*\brc{E_i}\Sub{i\in I}}{\brc{E_i}\Sub{i\in
  I}}$ constructed in \cref{rem:cartesian-maps} is indeed cartesian in
  $\Mor[p]{\FAM{\CCat}}{\SET}$.
\end{exercise}

In fact, we can strengthen \cref{exc:cart-lift-in-fam} to give an extrinsic
characterization of cartesian morphisms in $\FAM{\CCat}$: cartesian morphisms
are exactly the ``fiberwise isomorphisms''.

\begin{exercise}\label{exc:cart-map-in-fam}
  Let $E = \brc{E_i}\Sub{i\in I}$ and $F = \brc{F_j}\Sub{j\in J}$ be objects of
  $\FAM{\CCat}$, and let $\Mor[f]{E}{F}$ be a morphism
  between them. Show that the following are equivalent:
  \begin{enumerate}
    \item the morphism $\Mor[f]{E}{F}$ is cartesian;
    \item for each $i\in I$, the component $\Mor[f_i]{E_i}{F_{pfi}}$ is an isomorphism.
  \end{enumerate}
\end{exercise}

The name ``cartesian morphism'' is inspired by pullbacks, as we see in \cref{ex:cod-pullback} below.

\begin{exercise}\label[exercise]{ex:cod-pullback}
  Let $\BCat\Sup{\to}$ be the \emph{arrow category} of $\BCat$, whose objects
  are morphisms of $\BCat$ and whose morphisms are commuting squares between
  them; let $\Mor[\Con{cod}]{\BCat\Sup{\to}}{\BCat}$ be the \emph{codomain
  functor} that projects the codomain of a map $\Mor{A}{B}$. Show that a
  morphism $\Mor{E}{F}\in\BCat\Sup{\to}$ is cartesian if and only if the
  corresponding square in $\BCat$ is a pullback square (also called a cartesian square).
\end{exercise}

The relationship between cartesian morphisms and pullback squares expressed by
\cref{ex:cod-pullback} suggests a generalization of the conventional ``pullback
corners'' notation to an arbitrary fibration, which we shall use liberally.

\begin{convention}[Generalized ``pullback corners'']
  Let $\Mor[p]{\ECat}{\BCat}$ be a functor; when we wish to indicate that a
  morphism $\Mor{E}{F}$ in $\ECat$ is cartesian over $\Mor{pE}{pF}$, we will
  display it using a ``pullback corner'' notation as follows:
  \[
    \DiagramSquare{
      nw/style = pullback,
      nw = E,
      ne = F,
      sw = pE,
      se = pF,
      west/style = {|->},
      east/style = {|->},
      height = 1.5cm,
    }
  \]
\end{convention}

Finally we may give the definition of a fibration.

\begin{definition}
  A functor $\Mor[p]{\ECat}{\BCat}$ is called a \textbf{fibration} when for any
  object $E\in \ECat$ and morphism $\Mor{B}{pE} \in \BCat$ there exists a
  cartesian morphism $\Mor{H}{E}$ lying over $\Mor{B}{pE}$. The cartesian
  morphism is often called the \textbf{cartesian lift} of $\Mor{B}{pE}$ at $E$.
\end{definition}

\begin{convention}
  We will depict a fibration $\FibMor[p]{\ECat}{\BCat}$ using triangular
  arrows. When we wish to leave the functor implicit, we refer to $\ECat$ as a
  \textbf{fibered category} over $\BCat$. In the same way that one writes
  $A\times B$ for the apex of a product diagram, we will often write
  $\Mor[\bar{u}]{u^*E}{E}$ for the cartesian lift of $\Mor[u]{B}{pE}$ at $E$ as
  depicted below:
  \[
    \DiagramSquare{
      height = 1.5cm,
      west/style = lies over,
      east/style = lies over,
      ne = E,
      se = pE,
      sw = B,
      nw = u^*E,
      north = \bar{u},
      south = u,
      nw/style = pullback,
    }
  \]
\end{convention}

In the case of $\FAM{\ECat}$, the existence of cartesian lifts for each
$\Mor[u]{B}{pE}$ corresponds to the \emph{reindexing} functors
$\Mor[u^*]{\FAM{\CCat}\Sub{pE}}{\FAM{\CCat}\Sub{B}}$.

\begin{exercise}
  Verify that the functor $\FibMor{\FAM{\CCat}}{\SET}$ is a fibration.
\end{exercise}

\begin{exercise}
  Conclude from \cref{ex:cod-pullback} that the codomain functor
  $\Mor{\BCat\Sup{\to}}{\BCat}$ is a fibration if and only if $\BCat$ has all
  pullbacks.
\end{exercise}

When the codomain functor $\Mor{\BCat\Sup{\to}}{\BCat}$ is a fibration, we will
refer to it as the \emph{fundamental fibration}, written $\FFib{\BCat}$
following \citet{streicher:2021:fib}.

\subsubsection{Small categories, internal categories}

An ordinary category need not have a set of objects --- for instance, the
category $\SET$ of all sets has a \emph{proper class} of objects.\footnote{In
this paper, we are somewhat agnostic about set theoretic foundations. Our
discussion is compatible with the viewpoint of ZFC, in which classes are taken
to be formulas at the meta-level; our discussion is, however, also compatible
with other accounts of the ``set--class'' distinction, such as NBG set theory,
MK set theory, or the universe-based approaches of Grothendieck~\citep{sga:4}
and \citet{maclane:1998}.} Likewise, it is possible to find categories such
that between two objects there may be a proper class of morphisms (\eg the
category of sets and isomorphism classes of spans between them). A category
that has hom \emph{sets} is called \emph{locally small}, and we will refer to a
category that has, up to equivalence, a \emph{set} of objects as \emph{globally
small}. A category that has both these properties is equivalent to a
\emph{small category} in the ordinary sense.
Small categories are very useful: for instance, if $\mathbb{C}$ is a small
category then the category of functors $\brk{\mathbb{C},\SET}$ is a
Grothendieck topos.  Functor categories of this kind play an important role in
theoretical computer science~\citep[\eg][]{oles:1986,reynolds:1995,bmss:2011}.

The idea of a (globally, locally) small category can be relativized from $\SET$
to another category in two \emph{a priori} different ways that ultimately
coincide up to equivalence. The simplest and move na\"ive way to think of a
small category $\mathbb{C}$ in a category $\BCat$ is as an \emph{algebra} for
the sorts and operations of the \emph{theory of a category} internal to
$\BCat$, which we develop below; the more sophisticated way is to view
$\mathbb{C}$ as a fibration over $\BCat$ satisfying a generalization of the
global and local smallness conditions.

\begin{definition}
  Let $\BCat$ be a category that has pullbacks. An \textbf{internal category}
  or \textbf{category object} in $\BCat$ is given by:
  \begin{enumerate}
    \item an object $\mathbb{C}_0\in\BCat$ of \emph{objects},
    \item and an object $\mathbb{C}_1\in\BCat$ of \emph{morphisms},
    \item and source and target maps $\Mor[s,t]{\mathbb{C}_1}{\mathbb{C}_0}$,
    \item and a morphism $\Mor[i]{\mathbb{C}_0}{\mathbb{C}_1}$ choosing the identity maps, such that $s\circ i = \Idn{\mathbb{C}_0} = t\circ i$,
    \item and a morphism $\Mor[c]{\mathbb{C}_1\times\Sub{\mathbb{C}_0}\mathbb{C}_1}{\mathbb{C}_1}$ choosing composite maps such that $s\circ c = s\circ \pi_1$ and $t\circ c = t\circ \pi_2$,
    \item (plus associativity and unit laws for composition and identity)
  \end{enumerate}
\end{definition}

\begin{observation}
  When $\BCat = \SET$ we obtain exactly the ordinary notion of a small
  category, \ie a small category is the same thing as an
  \emph{internal category} or \emph{category object} in $\SET$.
\end{observation}

In the previous section, we have argued that fibrations are a fruitful way to
think about categories defined relative to another category. Indeed, we may
view an internal category $\mathbb{C}$ as a fibration via a process called
\emph{externalization}. This proceeds in two steps; first we construct a
\emph{presheaf of categories} on $\BCat$, and then we use the
\emph{Grothendieck construction} to turn it into a fibration.

\begin{construction}[The presheaf of categories associated to an internal category]\label[construction]{con:icat-to-psh}
  Let $\mathbb{C}$ be an internal category in $\BCat$. We may define a presheaf
  of categories $\Mor[\mathbb{C}^\bullet]{\OpCat{\BCat}}{\CAT}$ like so:
  \begin{enumerate}
    \item for $I\in\BCat$, an object of $\mathbb{C}^I$ is given by a morphism $\Mor[\alpha]{I}{\mathbb{C}_0}$,
    \item for $I\in\BCat$, a morphism $\Mor{\alpha}{\beta}\in\mathbb{C}^I$ is given by a morphism $\Mor[h]{I}{\mathbb{C}_1}$ such that $s\circ h = \alpha$ and $t\circ h = \beta$,
    \item for $\Mor[u]{J}{I}$ in $\BCat$, the reindexing $\Mor[u^*]{\mathbb{C}^I}{\mathbb{C}^J}$ is given on both objects and morphisms by precomposition with $u$.
  \end{enumerate}
\end{construction}

\begin{construction}[The Grothendieck construction]
  Let $\Mor[\mathbb{C}^\bullet]{\OpCat{\BCat}}{\CAT}$ be a presheaf of categories; we define its
  \emph{total category} $\int\Sub{\BCat}\!\mathbb{C}^\bullet$ as follows:
  \begin{enumerate}
    \item an object of $\int\Sub{\BCat}\!\mathbb{C}^\bullet$ is given by a pair of an object $I\in\BCat$ and an object $c\in \mathbb{C}^I$,
    \item a morphism $\Mor{\prn{J,c}}{\prn{I,d}}$ is given by a pair of a morphism $\Mor[u]{J}{I}\in\BCat$ and a morphism $\Mor{c}{\mathbb{C}^ud}$ in $\mathbb{C}^J$.
  \end{enumerate}

  There is an evident functor
  $\Mor[p]{\int\Sub{\BCat}\!\mathbb{C}^\bullet}{\BCat}$; it is this functor
  that is referred to as the \emph{Grothendieck construction}.
\end{construction}

\begin{exercise}
  Verify that the Grothendieck construction of any presheaf of categories
  $\Mor[\mathbb{C}^\bullet]{\OpCat{\BCat}}{\CAT}$ is a fibration.
\end{exercise}

Observe that the presheaf of categories associated to an internal category is,
in each fiber $I$, the category object in $\SET$ obtained by restricting along
the functor $\Mor[I]{\Kwd{1}}{\BCat}$.

\begin{definition}[Externalization of an internal category]
  Let $\mathbb{C}$ be an internal category in $\BCat$; its
  \emph{externalization} is defined to be the Grothendieck construction
  $\FibMor{\brk{\mathbb{C}}\coloneqq
  \int\Sub{\BCat}\!\mathbb{C}^\bullet}{\BCat}$ of the associated presheaf of
  categories (\cref{con:icat-to-psh}).
\end{definition}

\begin{remark}
  When $\mathbb{C}$ is an internal category in $\SET$, the externalization $\brk{\mathbb{C}}$ is the \emph{family fibration} $\FAM{\mathbb{C}}$ described in \cref{con:family-fibration}.
\end{remark}

As promised we may now isolate the properties of the fibered category
$\brk{\mathbb{C}}$ that correspond (up to equivalence) to arising by
externalization from an internal category.

\begin{definition}\label[definition]{def:globally-small}
  A fibered category $\FibMor[p]{\ECat}{\BCat}$ is called \emph{globally small} if
  there is an object $T\in\ECat$ such that  for any $X\in\ECat$ there exists a
  (not necessarily unique) cartesian map $\Mor{X}{T}$.
\end{definition}

The property of global smallness described above is often phrased as
$\FibMor[p]{\ECat}{\BCat}$ having a ``generic object'' $T\in\ECat$, but the reason
for this paper's existence is that the precise definition of ``generic'' means
in this context is somewhat controversial, and the bulk of the present paper is
devoted to justifying the precise meaning for genericity that we have chosen
(in agreement with \cref{def:globally-small}). Note that our
\cref{def:globally-small} does not agree with that of \citet{jacobs:1999}, who
assumes additional (somewhat rare) properties of the object $T$, namely that
for $X\in\ECat$ there is exactly one morphism $\Mor{pX}{pT}$ lying underneath a
cartesian morphism $\Mor{X}{T}$. We will show \cref{sec:consequences:icat} that
ours is the correct definition.

\begin{definition}[\citet{benabou:1975,streicher:2021:fib}]\label[definition]{def:locally-small}
  A fibered category $\FibMor[p]{\ECat}{\BCat}$ is called \emph{locally small}
  when for any $I\in \BCat$ and $X,Y\in\ECat_I$ there exists a span
  $X\xleftarrow{f} H \xrightarrow{g} Y$ with $\Mor[f]{H}{X}$ cartesian over
  $\Mor[pg]{H}{Y}$ as depicted below,
  \[
    \begin{tikzpicture}[diagram]
      \SpliceDiagramSquare<l/>{
        height = 1.5cm,
        north/style = <-,
        south/style = <-,
        west/style = lies over,
        east/style = lies over,
        ne = H,
        se = pH,
        nw = X,
        sw = I,
        south = pg,
        north = f,
        ne/style = ne pullback,
      }
      \SpliceDiagramSquare<r/>{
        height = 1.5cm,
        glue = west,
        glue target = l/,
        east/style = lies over,
        ne = Y,
        se = I,
        north = g,
        south = pg,
      }
    \end{tikzpicture}
  \]
  such that for any other span $X\xleftarrow{f'} K \xrightarrow{g'} Y$ where
  $\Mor[f']{K}{X}$ is cartesian over $\Mor[pg']{pK}{pY}$,
  there is a unique map $\Mor{K}{H}$ making the the following diagram commute:
  
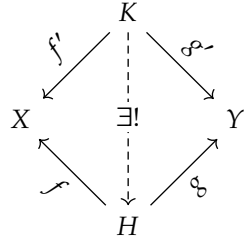
\[
    \begin{tikzpicture}[diagram]
      \node (K) {$K$};
      \node (X) [below left = of K] {$X$};
      \node (Y) [below right = of K] {$Y$};
      \node (H) [below left = of Y] {$H$};
      \draw[->] (K) to node[sloped,above] {$f'$} (X);
      \draw[->] (K) to node[sloped,above] {$g'$} (Y);
      \draw[->] (H) to node[sloped,below] {$g$} (Y);
      \draw[->] (H) to node[sloped,below] {$f$} (X);
      \draw[exists,->] (K) to node [upright desc] {$\exists!$} (H);
    \end{tikzpicture}
  \]
\end{definition}

\begin{exercise}[Intermediate]
  Let $\CCat$ be an ordinary category; show that the fibered category
  $\FibMor{\FAM{\CCat}}{\SET}$ is locally small if and only if $\CCat$ is
  locally small. Show that $\CCat$ is equivalent to a small category if and
  only if $\FibMor{\FAM{\CCat}}{\SET}$ is both globally small and locally small.
\end{exercise}

One of the fundamental results of fibered category theory is that, up to
equivalence, global and local smallness in the sense of
\cref{def:globally-small,def:locally-small} suffice to detect internal
categories, which we recall from \citet{benabou:1975}:

\begin{proposition}[Th\'eor\`eme~2, \citet{benabou:1975}]\label[proposition]{prop:small-iff-globally-and-locally-small}
  A fibration $\FibMor{\ECat}{\BCat}$ is equivalent to the externalization of
  an internal category $\mathbb{E}$ if and only if it is both globally and locally small.
\end{proposition}

Although we do not include the (standard) proof of
\cref{prop:small-iff-globally-and-locally-small}, it is instructive to
understand the object $T\in\brk{\mathbb{E}}$ in the externalization of an
internal category $\mathbb{E}$ that renders $\brk{\mathbb{E}}$ globally small.
Recalling the definition of the externalization via the Grothendieck
construction, we define the \textbf{weak generic} object $T$ to be the pair
$\prn{\mathbb{E}_0, \Mor[\Idn{\mathbb{E}_0}]{\mathbb{E}_0}{\mathbb{E}_0}}$
given by the object of objects and its identity map.

\begin{example}\label{ex:weak-generic-fam}
  When $\mathbb{E}$ is a
  small category, then the \textbf{weak generic} object of
  $\brk{\mathbb{E}}=\FAM{\mathbb{E}}$ can be written as the family
  $\brc{E}{\Sub{E\in \mathbb{E}_0}}$.
\end{example}

\subsubsection{Cleavages and splittings}

We briefly recall the definitions of cleavages and splittings for a fibration,
as they are relevant to the rest of this paper.

\begin{definition}
  A \emph{cleavage} for a fibration $\FibMor[p]{\ECat}{\BCat}$ is a choice
  $\mathfrak{r}$ of cartesian liftings, assigning to each morphism $\Mor[u]{I}{pX}\in
  \BCat$ in the base an object $\mathfrak{r}_u X\in \ECat$ over $I$ and a cartesian
  morphism $\Mor[\overline{\mathfrak{r}_u}]{\mathfrak{r}_u X}{pX}$ over $u$.
\end{definition}

The data of a cleavage $\mathfrak{r}$ extends for each $\Mor[u]{I}{J}\in \BCat$
to a reindexing functor $\Mor[\mathfrak{r}_u]{\ECat\Sup{J}}{\ECat\Sup{I}}$.

\begin{definition}
  A cleavage $\mathfrak{r}$ for a fibration $\FibMor[p]{\ECat}{\BCat}$ is
  called \emph{split} when the assignment of reindexing functors $u\mapsto
  \mathfrak{r}_u$ strictly preserves identities and compositions.
\end{definition}

\begin{convention}
  When the cleavage $\mathfrak{r}$ is understood, we will often write $u^*$ for
  $\mathfrak{r}_u$.
\end{convention}

A fibration equipped with a cleavage is called a \emph{cloven fibration}; we
may use the axiom of choice to equip any fibration with a (non-canonical)
cleavage.  When the cleavage associated to a cloven fibration is split, we
speak of \emph{split fibrations}.
Observe that a splitting allows one to view a fibration $\FibMor{\ECat}{\BCat}$
as a presheaf of categories $\Mor{\OpCat{\BCat}}{\CAT}$ sending each
$I\in\BCat$ to the fiber $\ECat^I$. Recalling \cref{con:icat-to-psh}, we see
that the externalization of any internal category is split.

\subsection{Goals and structure of this paper}

Although the property stated in \cref{def:globally-small}, that there exist
cartesian morphisms $\Mor{X}{T}$ for any $X$, is the most that can be required
for an arbitrary internal category, more restrictive notions of generic object
have proved important in practice for different applications. Unfortunately,
over the years a number of competing definitions have proliferated throughout
the literature --- and some of the more established of these definitions lead
to false conclusions when taken too literally, as we point out in
\cref{sec:consequences:icat} in our discussion of Jacobs' mistaken
Corollary~9.5.6.

The goal of this paper, therefore, is to argue for a new alignment of
terminology for the different forms of generic object that is both internally
consistent \emph{and} reflects the use of generic objects in practice. Because
generic objects play an important role in several areas of application
(categorical logic, algebraic set theory, homotopy type theory, denotational
semantics of polymorphism, \etc), we believe that we have sufficient evidence
today to correctly draw the map.

\begin{itemize}

  \item In \cref{sec:kinds-of-generic-object}, we recall the definitions of
    several variants of generic object by \citet{jacobs:1999}; our main
    observation is that a \textbf{split generic} object in the sense of \opcit
    need not be a \textbf{generic} object in the same sense.

  \item In \cref{sec:consequences}, we analyze the consequences of the
    definitions discussed in \cref{sec:kinds-of-generic-object} for the use of
    generic objects in internal category theory, tripos theory, denotational
    semantics of polymorphism, algebraic set theory, and homotopy type theory.

  \item In \cref{sec:proposal}, we propose new unified terminology and
    definitions for all extant forms of generic object (as well as one
    \emph{new} one). Our proposal is summarized and compared with the
    literature in \cref{tab:rosetta-stone}.

\end{itemize}

\section{Four kinds of generic object}\label{sec:kinds-of-generic-object}

We begin by recalling Definition~5.2.8 of \citet{jacobs:1999}, from which we
omit some additional characterizations that will not play a role in our analysis.
\begin{quote}
  \itshape
  Consider a fibration $\FibMor[p]{\ECat}{\BCat}$ and an object $T$ in the total category $\ECat$. We call $T$ a
  \begin{enumerate}[i)]
    \item \textbf{weak generic} object if
      $
        \forall X\in\ECat. \exists \Mor[f]{X}{T}. f\text{ is cartesian}
      $.

    \item \textbf{generic} object if
      $
        \forall X\in\ECat. \exists! \Mor[u]{pX}{pT}. \exists \Mor[f]{X}{T}. f\text{ is cartesian over $u$}
      $.

    \item \textbf{strong generic} object if
      $
        \forall X\in\ECat. \exists! \Mor[f]{X}{T}. f\text{ is cartesian}
      $.
  \end{enumerate}
\end{quote}

\citet{jacobs:1999} then defines \textbf{split generic} objects for split
fibrations in Definition~5.2.1, paraphrased below:
\begin{quote}
  A split fibration $\FibMor[p]{\ECat}{\BCat}$ has a \textbf{split generic
  object} if there is an object $\Omega\in\BCat$ together with natural
  isomorphism $\Mor[\theta]{\BCat\prn{-,\Omega}}{\Ob\,{\ECat_\bullet}}$ in
  $\brk{\OpCat{\BCat},\SET}$, where the presheaf $\Ob\,\ECat_\bullet$ is
  defined using the splitting.
\end{quote}

A useful characterization of \textbf{split generic} objects is given in
Lemma~5.2.2 of \opcit:
\begin{quote}
  \itshape
  A split fibration $\FibMor[p]{\ECat}{\BCat}$ has a \textbf{split generic} object if and only if there is an object $T\in\ECat$ with the property that $\forall X\in \ECat.\exists!\Mor[u]{pX}{pT}. u^*T = X$.~\citep{jacobs:1999}
\end{quote}

\begin{scholium}
  The \textbf{weak} and \textbf{strong generic objects} of \citet{jacobs:1999}
  are referred to by \citet{phoa:1992} as \emph{generic objects} and
  \emph{skeletal generic objects}. \citet{phoa:1992} does not consider the
  intermediate notion.
  On the other hand, \citet{phoa:1992} defines \emph{strict generic objects}
  relative to an arbitrary (non-split) cleavage; a \textbf{split generic
  object} is indeed a \emph{strict generic object} in the sense of Phoa, but
  even for a split cleavage, a \emph{strict generic object} need not be a
  \textbf{split generic} object. We will discuss Phoa's terminology more in
  \cref{schol:skeletal}.
\end{scholium}

\subsection{Separating generic objects from strong generic objects}

\citet{jacobs:1999} notes that generic and strong generic objects coincide in
fibered preorders, but they may differ otherwise --- the difference emanating
from the presence of non-trivial vertical automorphisms.

\NewDocumentCommand\Sk{m}{\Kwd{sk}\,#1}

\begin{example}\label{ex:delooping}
  Let $G$ be a group containing two distinct elements $u\not= v$, and let
  $\Deloop{G}$ be the groupoid with a single object whose hom set is $G$
  itself. The family fibration $\FibMor{\FAM{\Deloop{G}}}{\SET}$ is
  a fibered category whose objects are just sets, but such that a morphism
  $\Mor[f]{I}{J}\in \FAM{\Deloop{G}}$ is a pair $\prn{f,x}$ of a
  function $\Mor[f]{I}{J}\in \SET$ together with a generalized element
  $\Mor[x]{I}{G}$. Moreover, every morphism in $\FAM{\Deloop{G}}$
  is cartesian as $\Deloop{G}$ is a groupoid. The unique
  object $T\in \FAM{\Deloop{G}}\Sub{1\Sub{\SET}}$ is clearly
  \textbf{generic}, but not \textbf{strong generic}. Indeed, we have two
  distinct cartesian morphisms $\Mor{T}{T}$ given by the pairs
  $\prn{\Idn{1\Sub{\SET}},u}\not=\prn{\Idn{1\Sub{\SET}},v}$.
\end{example}

\begin{example}\label{ex:skeleton}
  Another class of examples comes from considering skeleta of full subcategories
  of $\SET$. For instance, one may take the skeleton of a Grothendieck universe
  and then externalize to obtain a fibration that has a \textbf{generic} object
  that is not \textbf{strong}.
\end{example}

\subsection{A split generic object need not be a generic object}\label{sec:split-generic-not-generic}

\begin{construction}[The canonical splitting of
  the externalization]\label[construction]{con:canonical-splitting}
  The externalization of an internal category is split in a canonical way:
  given $\prn{I,c}\in \brk{\mathbb{C}}$ and $\Mor[u]{J}{I}$, we choose
  $u^*\prn{I,c} = \prn{J, c\circ u}$. The cartesian morphism
  $\Mor{u^*\prn{I,c}}{\prn{I,c}}$ is given by the pair $\prn{u, \Idn{} \circ
  c\circ u}$ where $\Mor[\Idn{}]{\mathbb{C}_0}{\mathbb{C}_1}$ is the generic
  identity morphism.
\end{construction}

\begin{construction}[Weak, split generic objects in the externalization]\label[construction]{con:externalization-weak-generic-object}
  The externalization $\FibMor[p]{\brk{\mathbb{C}}}{\BCat}$ has a \textbf{weak
  generic object} $T = \prn{\mathbb{C}_0, \Idn{\mathbb{C}_0}}$. Relative to the
  splitting of $\brk{\mathbb{C}}$ from \cref{con:canonical-splitting}, the
  \textbf{weak generic} object $T$ is also a \textbf{split generic} object.
\end{construction}

The \textbf{weak generic} object of the externalization of an internal category
defined in \cref{con:externalization-weak-generic-object} obviously need not be
a \textbf{strong generic} object, but it may be more surprising to learn that
it also need not be a \textbf{generic} object at all. This can happen, for
instance, when the internal category $\mathbb{C}$ has two distinct isomorphic
objects; the following concrete example illustrates the problem:

\begin{example}[A split generic object that is not a generic object]
  Let $U$ be a set of sets containing two distinct elements $A,B$ with the same
  cardinality, and $\SET_U\subseteq\SET$ be the full subcategory of $\SET$
  spanned by $U$.  Then the family fibration
  $\FibMor{\FAM{\SET_U}}{\SET}$ has a \textbf{split generic} object
  $T$ given by the pair $T=\prn{U, \Idn{U}}$, but $T$ is nonetheless not a
  \textbf{generic} object. Indeed, we have two cartesian morphisms
  $\Mor{\prn{1\Sub{\SET}, A}}{T}$ lying over distinct elements
  $\Mor[A\not=B]{1\Sub{\SET}}{U}$ respectively.
\end{example}

\begin{corollary}
  A split generic object is not necessarily a generic object.
\end{corollary}

\subsection{Skeletal and gaunt small categories}

We recall from \cref{ex:weak-generic-fam} that for a small category $\mathbb{C}$,
the family fibration $\FAM{\mathbb{C}}$ has a \textbf{weak generic} object $T =
\brc{C}\Sub{C\in\mathbb{C}_0}$. We will relate properties of the category
$\mathbb{C}$ to corresponding properties of the \textbf{weak generic} object
$T$.

\begin{definition}\label[definition]{def:skeletal}
  A category is called \emph{skeletal} when any to isomorphic objects are equal.
\end{definition}

\begin{definition}\label[definition]{def:gaunt}
  A category is called \emph{gaunt} when any isomorphism in that category is the identity.
\end{definition}

\begin{scholium}
  The term \emph{gaunt} is used by
  \citet{barwick-schommer-pries:2011,nlab:gaunt-category,hottbook}. In passing,
  \citet{johnstone:2002} has referred to such categories as \emph{stiff}.
\end{scholium}

\begin{lemma}\label[lemma]{lem:gaunt-skeletal-generic-object-comparison}
  The category $\mathbb{C}$ is (respectively skeletal, gaunt) if and only if
  $T$ is (respectively \textbf{generic}, \textbf{strong generic}).
\end{lemma}

\begin{proof}
  (1a) If $\mathbb{C}$ is skeletal, then $T$ is \textbf{generic}; fix any
  family $\brc{D_i}\Sub{i\in I}$ and cartesian map
  $\Mor[\chi]{\brc{D_i}\Sub{i\in I}}{\brc{C}{\Sub{C\in\mathbb{C}}}}$.
  By \cref{exc:cart-map-in-fam}, $p\chi$ sends each $i\in I$ to an
  object $p\chi_i\in\mathbb{C}$ that is isomorphic to $D_i$; as
  $\mathbb{C}$ is skeletal, it follows that $p\chi_i = D_i$ and so we
  conclude that $T$ is \textbf{generic}.

  (1b) Assume conversely that $T$ is \textbf{generic} and fix an
  isomorphism $f:D_0\cong D_1$ in $\mathbb{C}$. We have two cartesian maps
  $\Mor[h_0,h_1]{\brc{D_0}}{T}$, with one lying over
  $\Mor[D_0]{\brc{*}}{\mathbb{C}_0}$ via the identity morphism and the
  other lying over $\Mor[D_1]{\brc{*}}{\mathbb{C}_0}$ via $f$. Since $T$
  is \textbf{generic}, these two cartesian maps must lie over the same
  element of $\mathbb{C}_0$, so we have $D_0=D_1$.

  (2a) If $\mathbb{C}$ is gaunt, then $T$ is \textbf{strong generic}; fix any
  two cartesian morphisms $\Mor[h_0,h_1]{\brc{D_i}\Sub{i\in
  I}}{\brc{C}\Sub{C\in\mathbb{C}}}$. Because $\mathbb{C}$ is gaunt and thus
  skeletal, we know that $ph_0 = ph_1$; thus $h_0$ assigns to each $i\in I$ an
  isomorphism $\Mor[h_{0,i}]{D_i}{ph_1i}$ which is (by assumption) necessarily
  the identity. Thus both $h_0$ and $h_1$ must send every $i\in I$ to the
  identity map on $D_i$ and are thus equal.

  (2b) Conversely we assume that $T$ is \textbf{strong generic} to check that any
  isomorphism in $\mathbb{C}$ is an identity map; since $T$ is
  necessarily also \textbf{generic}, it follows by the first case of the
  present lemma that we may consider just the automorphisms in
  $\mathbb{C}$, considering
  \cref{lem:gaunt-iff-ess-gaunt-and-skeletal}. To show that any
  automorphism $f:D\cong D$ in $\mathbb{C}$ is the identity morphism,
  we proceed exactly as in the proof of
  \cref{lem:strong-generic-essentially-gaunt} by observing that the two
  cartesian maps corresponding to the identity map and the automorphism
  $f$ respectively are necessarily equal by our assumption that $T$ is
  \textbf{strong generic}.
\end{proof}

\subsection{Generic objects from weak generic objects}\label{sec:gen-from-split-gen}

Although the externalization of an internal category necessarily has a (split)
\textbf{weak generic} object, in some cases it
may also have a \textbf{generic} object $T'$ that embodies the
\emph{skeleton} of $\mathbb{C}$ as in \cref{ex:skeleton}, but $T'$ is usually
different from the $T$.

\begin{construction}[Computing the skeleton of a small category]
  Suppose that $\BCat = \SET$ and thus $\mathbb{C}$ is an ordinary small
  category. Then we may consider the quotient $\mathbb{C}_0/\cong$ of the
  objects of $\mathbb{C}$ under isomorphism; in other words, this is the set of
  isomorphism classes of $\mathbb{C}$-objects. Using the axiom of choice, we
  may arbitrarily choose a section $\Mor[s]{\mathbb{C}_0/\cong}{\mathbb{C}_0}$
  to the quotient map; moreover, we may choose a function associating to each
  $u\in\mathbb{C}_0$ an isomorphism $u\cong s\brk{u}\Sub{/\cong}$.
\end{construction}

\begin{lemma}
  The pair $T' = \prn{\mathbb{C}_0/\cong, s}$ is a \textbf{generic} object for $\FAM{\mathbb{C}}$.
\end{lemma}

\begin{proof}
  Fixing $\prn{I,c}\in\FAM{\mathbb{C}}$ we must choose a unique $\Mor[u]{I}{pT'}$
  such that there exists a cartesian map $\Mor{\prn{I,c}}{T}$ lying over $u$. We
  choose $u\prn{i} = \brk{c\prn{i}}\Sub{/\cong}$, taking each index $i\in I$ to
  the isomorphism class of $c\prn{i}$.

  \begin{enumerate}

    \item First of all, it is clear that there exists a cartesian map lying over
      $u$ in the correct configuration.

    \item Fixing $\Mor[v]{I}{pT'}$ such that there exists a cartesian map
      $\Mor{\prn{I,c}}{T}$ lying over $v$, it remains to show that $v = u$. This
      follows because such a cartesian map ensures that $v$ and $u$ are the same
      family of isomorphism classes of objects.
      \qedhere

  \end{enumerate}
\end{proof}

\begin{lemma}
  If the $T'$ defined above is a \textbf{split generic} object, then $\mathbb{C}$ is skeletal.
\end{lemma}

\begin{proof}
  We have already seen that $T$ is a \textbf{split generic} object, hence if
  $T'$ is also \textbf{split generic} we have isomorphisms ${pT}\cong
  \mathbb{C}_0 \cong pT'$ and so we have $\mathbb{C}_0\cong
  \mathbb{C}_0/\cong$.
\end{proof}

We finally observe in \cref{lem:strong-generic-essentially-gaunt} that if $T'$
is \textbf{strong generic}, then $\mathbb{C}$ is \emph{essentially gaunt} in
the sense defined below.

\begin{definition}\label[definition]{def:ess-gaunt}
  A category is called \emph{essentially gaunt} when it is equivalent to a
  gaunt category.
\end{definition}

Simon Henry has made the following observation:\footnote{Comment on the
MathOverflow thread entitled \emph{Name for `Category without nontrivial
automorphisms'?}, \url{https://mathoverflow.net/q/370764}.}

\begin{observation}
  A category is essentially gaunt if and only if any \emph{automorphism} in that category is the identity.
\end{observation}

\begin{lemma}\label[lemma]{lem:gaunt-iff-ess-gaunt-and-skeletal}
  A category is gaunt if and only if it is skeletal and essentially gaunt.
\end{lemma}

\begin{proof}
  If $\mathbb{C}$ be gaunt, it is obviously both essentially gaunt and
  skeletal.  Conversely, if $\mathbb{C}$ is essentially gaunt and skeletal,
  given any isomorphism $f:D\cong C$ we have $D=C$ and this $f$ is an
  automorphism, which (by gauntness) is the identity.
\end{proof}

\begin{lemma}\label[lemma]{lem:strong-generic-essentially-gaunt}
  If $T'$ as defined above is a \textbf{strong generic} object, then
  $\mathbb{C}$ is essentially gaunt in the sense of \cref{def:ess-gaunt}.
\end{lemma}

\begin{proof}
  Let $\Mor[f]{c}{c}$ be an automorphism in $\mathbb{C}$, \ie a vertical
  isomorphism in $\FAM{\mathbb{C}}\Sub{1\Sub{\SET}}$. Since $T'$ is \textbf{strong
  generic}, there exists a unique cartesian morphism
  $\Mor{\prn{1\Sub{\SET},c}}{T'}$; this means that there is a unique element
  $\brk{c}\in \mathbb{C}_0/\cong$ and a unique isomorphism
  $\Mor[h]{c}{s\brk{c}}$. Writing $\phi_c : c\cong s\brk{c}$ for the (globally)
  chosen isomorphism, we have $f;\phi_c = h = \phi_c$ and hence $f = \Idn{c}$,
  so $\mathbb{C}$ is essentially gaunt.
\end{proof}

Later on in \cref{sec:proposal}, we will see that the original \textbf{weak
generic} object $T$ being \textbf{strong generic} corresponds to $\mathbb{C}$
being gaunt.

Thus we conclude that although the family fibration $\FAM{\mathbb{C}}$ over
$\SET$ of a small category $\mathbb{C}$ does have a generic object, this
generic object cannot be either a \textbf{split generic} object or a
\textbf{strong generic} object except in somewhat contrived scenarios.

\subsection{Weak generic objects are the correct generalization of split generic objects}\label{sec:weak-and-split}

It is clear that any split generic object is in particular a weak generic object; but the
converse \emph{also} holds in a certain sense that we make precise below.

\begin{construction}[Presheaf of categories]\label{con:pr-of-cat}
  Let $\FibMor[p]{\ECat}{\BCat}$ be a fibered category equipped with a cleavage $\mathfrak{r}$, and let $T\in\ECat$ be a
  \textbf{weak generic} object for $T$. We may construct a presheaf of categories
  $\Mor[\ECat^\bullet]{\OpCat{\BCat}}{\CAT}$ like so:
  \begin{enumerate}
    \item an object of $\ECat^I$ is a morphism $\Mor[\alpha]{I}{pT}$,
    \item a morphism $\Mor{\alpha}{\beta}$ in $\ECat^I$ is a vertical map $\Mor[h]{\mathfrak{r}_\alpha T}{\mathfrak{r}_\beta T}$ over $I$.
  \end{enumerate}
\end{construction}

\begin{construction}
  Let $\Mor[\mathbb{E}]{\OpCat{\BCat}}{\CAT}$ be a presheaf of categories; then
  the Grothendieck construction $\FibMor[q]{\int\!\mathbb{E}}{\BCat}$ has a
  splitting.  Given $\prn{J, \alpha}\in \int\!\ECat$ and $\Mor[u]{I}{J}$, the
  object $u^*\prn{J,\alpha}$ is chosen to be $\prn{I, \alpha\circ u}$ and the
  cartesian morphism $\Mor{u^*\prn{J,\alpha}}{\prn{J,\alpha}}$ over $u$ in
  $\int\!\ECat^\bullet$ is defined to be the pair $\prn{u, \Idn{u^*\alpha}}$.
\end{construction}

\begin{lemma}
  There is a fibered equivalence from $\int\!\ECat^\bullet$ to $\ECat$ over
  $\BCat$.
\end{lemma}

\begin{lemma}
  The pair $T' = \prn{T, \Idn{T}}$ is a \textbf{split generic} object in
  $\FibMor{\int\!\ECat^\bullet}{\BCat}$.
\end{lemma}

\begin{proof}
  Fixing $X\in \int\!\ECat^\bullet$, we must check that there exists a unique
  $\Mor[u]{pX}{pT'}$ such that $u^*T' = X$. Unfolding definitions, we fix $I\in
  \BCat$ and $\Mor[\alpha]{I}{T}$ to check that there is a unique
  $\Mor[u]{I}{T}$ such that $\prn{I, u} = \prn{I,\alpha}$. Of course, this is true with $u=\alpha$.
\end{proof}

Thus a \textbf{weak generic} object for a fibration $\FibMor{\ECat}{\BCat}$
generates in a canonical way a new \emph{equivalent} \textbf{split} fibration
$\FibMor{\ECat'}{\BCat}$ that has a \textbf{split generic object}. This is the
correspondence between \textbf{weak generic} and \textbf{split generic}
objects.
\section{Consequences for internal categories, tripos theory, polymorphism, \etc}\label{sec:consequences}

We recall several definitions from \citet{jacobs:1999} below in order to
illustrate a pattern.
\begin{enumerate}
    \itshape

  \item A \textbf{higher order fibration} is a first order
    fibration\footnote{We do not recall the definition of first order
    fibrations here, as the specifics are not important for the present paper.} with a
    \textbf{generic} object and a cartesian closed base category. Such a higher
    order fibration will be called \textbf{split} if the fibration is split and
    all of its fibered structure (including the \textbf{generic} object) is
    split.~\citep[Definition~5.3.1]{jacobs:1999}

  \item A \textbf{tripos} is a higher order fibration $\FibMor{\ECat}{\SET}$
    over $\SET$ for which the induced products $\Prod{u}$ and coproducts
    $\Coprod{u}$ along an arbitrary function $u$ satisfy the Beck-Chevalley
    condition.~\citep[Definition~5.3.3]{jacobs:1999}

  \item A \textbf{polymorphic fibration} is a fibration with a generic object,
    with fibered finite products and with finite products in its base category.
    It will be called \textbf{split} whenever all this structure is
    split.~\citep[Definition~8.4.1]{jacobs:1999}.

  \item A \textbf{polymorphic fibration} with $\Omega$ in the base as a generic object will be called
    \begin{enumerate}[(a)]
      \item a $\lambda{\to}$-fibration if it has fibered exponents;
      \item a $\lambda2$-fibration if it has fibered exponents and also simple $\Omega$-products and coproducts;
      \item a $\lambda\omega$-fibration if it has fibered exponents, simple products and coproducts, and exponents in its base category.~\citep[Definition~8.4.3]{jacobs:1999}
    \end{enumerate}

  \item Let $\Diamond$ be ${\to},2,\omega$. A $\lambda\Diamond$-fibration will be
    called \textbf{split} if all of its structure is split. In particular, its
    underlying polymorphic fibration is
    split.~\citep[Definition~8.4.4]{jacobs:1999}

\end{enumerate}

A \textbf{split generic} object need not be a \textbf{generic} object as we
have seen in \cref{sec:split-generic-not-generic}, and indeed, is only quite
rarely one. Consequently, a \textbf{split polymorphic fibration}
is not the same thing as a split fibration with a \textbf{split generic} object
and split fibered finite products.  This disorder suggests that a change of
definitions is in order, which we propose in \cref{sec:proposal}.

\subsection{Consequences for internal category theory}\label{sec:consequences:icat}

The somewhat chaotic status of \textbf{generic} objects \emph{vis-\`a-vis}
\textbf{split generic} objects has led to an erroneous claim by
\citet{jacobs:1999} that the externalization of an internal category has a
\textbf{generic} object. What is actually proved by \opcit in Lemma~7.3.2 is
that the externalization of an internal category has a \textbf{split generic} object, but this is later claimed (erroneously) to give rise to a
\textbf{generic} object in the proof of Corollary~9.5.6. Thus the claimed
result of Corollary~9.5.6, that a fibration is small (equivalent to the
externalization of an internal category) if and only if it is globally small
(has a \textbf{generic} object) and locally small, is mistaken. To give a
correct definition in the terminology of \opcit, one would define globally
small categories to assume a \textbf{weak generic} object, as we have done in
\cref{def:globally-small}.

\subsection{Consequences for tripos theory}

\citet{jacobs:1999} sketches in Example~5.3.4 the standard construction of a
tripos from a partial combinatory algebra, referring to
\citet{hyland-johnstone-pitts:1980} for several parts of the construction
(including the \textbf{generic} object); the standard definition of a tripos
involves a \textbf{weak generic} object
only~\citep{hyland-johnstone-pitts:1980,pitts:1981:thesis,pitts:2002}, a
discrepancy that has already been noted by \citet[p.~110]{birkedal:2000}.

\nocite{van-oosten:2008} This difference is quite destructive, as the
\textbf{weak generic} object of a realizability tripos will not generally be
\textbf{generic}, as pointed out by \citet[Theorem~4.6]{streicher:2017-2018}.
Thus we conclude that definition of tripos proposed by \citet{jacobs:1999}
rules out the main examples of naturally occurring triposes.

\begin{remark}
  A final subtlety: over the category of assemblies, the fibration of
  \emph{regular} subobjects does indeed have a \textbf{strong generic} object,
  but this is not part of the structure of the tripos.
\end{remark}

\subsection{Consequences for polymorphism}

The reception in the community studying polymorphism has been to either avoid
or tacitly correct the definition of generic object. For instance,
\citet{hermida:1993} speaks of \textbf{weak generic} objects and \textbf{strong
generic objects}, and gives the correct definition of
$\lambda{\to},\lambda2,\lambda\omega$ fibration in terms of \textbf{weak
generic objects}. On the other hand
\citet{birkedal-mogelberg-petersen:2005,johann-sojakova:2017,sojakova-johann:2018,
ghani-forsberg-orsanigo:2019} deal mainly with the \textbf{split generic
objects}, and thus do not seem to run into problems.
It can be seen, however, that the examples of (split)
$\lambda\Diamond$-fibrations in the cited works are \emph{not} in fact (split)
$\lambda\Diamond$-fibrations in the sense of \citet{jacobs:1999}, because they
do not have \textbf{generic} objects. Of course, the problem lies with the
definitions rather than the examples.

\subsection{Consequences for type theory and algebraic set theory}\label{sec:consequences:tt}

The idea of a \emph{universe} in a category has been abstracted from
Grothendieck's universes~\citep{sga:4} by way of the contributions of a number
of authors including
\citet{benabou:1973,martin-lof:itt:1975,street:1980:cosmoi,joyal-moerdijk:1995,hofmann-streicher:1997,streicher:2005}.
In fibered categorical language, a universe usually is a \emph{full
subfibration} that has a \textbf{weak generic} object --- potentially equipped
with additional structure.

In certain cases, the \textbf{weak generic} object of a universe is also a
\textbf{generic} object that is nonetheless not \textbf{strong}; one example is
Voevodsky's construction of a universe of well-ordered simplicial sets in the
context of the simplicial model of homotopy type theory, reported by
\citet{kapulkin-lumsdaine:2021}: it is \textbf{strong generic} for the fibered
category of well-ordered families of simplicial sets and order-preserving
morphisms, but only \textbf{generic} when considering morphisms that need not
preserve the well-orderings. On the other hand, universes of propositions, such as
the subobject classifier of a topos or the regular subobject classifier of a
quasitopos, are the main source of \textbf{strong generic objects} in nature.

The theory of universes as applied to \emph{type theory} on the one hand and
\emph{algebraic set theory} on the other hand motivates two additional variants
of generic object:

\begin{enumerate}

  \item In applications to type theory, it has been increasingly important for
    universes to have a \textbf{weak generic} object $T$ that satisfies an
    additional \emph{realignment} property relative to a class of monomorphisms
    $\mathcal{M}$; in particular, given a span of cartesian maps
    $U\leftarrowtail X \to T$ where $p\prn{U\leftarrowtail X}\in\mathcal{M}$,
    we need an extension $U\to T$ factoring $X\to T$ through $U\rightarrowtail
    X$. This realignment property has proved essential for the semantics of
    univalent universes in homotopy type
    theory~\citep{kapulkin-lumsdaine:2021,shulman:2015:elegant,shulman:2019,gratzer-shulman-sterling:2022:universes}
    as well as Sterling's \emph{synthetic Tait
    computability}~\citep{sterling:2021:thesis}, an abstraction of Artin gluing
    and logical relations.
    In \cref{sec:acyclic} we will discuss the generalization of the realignment
    property to an arbitrary fibered category as the \ul{acyclic generic}
    object, which lies strictly between \textbf{weak generic} and
    \textbf{generic} objects.

  \item In algebraic set theory~\citep{joyal-moerdijk:1995}, universes are considered that enjoy a form of
    generic object that is even weaker than \textbf{weak generic} object. For
    each $E\in\ECat$ one only ``locally'' has a morphism into $T$; the ``very
    weak'' generic objects of algebraic set theory seem to relate to
    \textbf{weak generic} objects in the same way that weak completeness
    relates to strong completeness in the context of
    stack completions~\citep{hyland:1988,hyland-robinson-rosolini:1990}, with important
    applications to the theory of polymorphism. As these very weak generic
    objects seem to play a fundamental role, we discuss them in more detail in
    \cref{sec:stacks} as part of our proposal for a new alignment of terminology.

\end{enumerate}
\section{A proposal for a new alignment of terminology}\label{sec:proposal}

Based on the data and experience of the applications of fibered categories in
the theory of triposes, polymorphism, and universes in type theory and
algebraic set theory, we may now proceed with some confidence to propose a new
alignment of terminology for generic objects that better reflects fibered
categorical practice. In this section, we distinguish our proposed usage from
that of other authors by \ul{underlining}.

We will define several notions of generic object in conceptual order rather
than in order of strength; in \cref{tab:rosetta-stone} we summarize our
terminology, and compare it to several representatives from the literature. In
\cref{fig:strength} we compare the strength of the different kinds of generic
object.

\begin{table}
  \begin{adjustbox}{center}
    \begin{tabular}{lllll}
      \toprule
      \textbf{Our Proposal} & \textbf{\citet{jacobs:1999}} & \textbf{\citet{phoa:1992}} & \textbf{\citet{hermida:1993}} & \textbf{\citet{streicher:2005}}\\
      \midrule
      \ul{weak generic} & --- & --- & --- & weak generic \\
      \ul{generic} & \textbf{weak generic} & generic & generic & generic\\
      \ul{acyclic generic} & --- & --- & --- & --- \\
      \ul{skeletal generic} & \textbf{generic} & --- & --- & --- \\
      \ul{gaunt generic} & \textbf{strong generic} & skeletal generic & strong generic & classifying \\
      \bottomrule
    \end{tabular}
  \end{adjustbox}

  \caption{A Rosetta stone for the different terminologies for generic objects.}
  \label{tab:rosetta-stone}
\end{table}

\begin{figure}
  \vspace{1em}
  \begin{center}
    \tikzset{box/.append style = {draw=black}, arrow/.append style = {line width=2pt, {stealth}-}}
    \begin{tikzpicture}
      \node[box] (weak) {\ul{weak generic}};
      \node[box,left=of weak] (plain) {\ul{generic}};
      \node[box,above left=of plain] (acyclic) {\ul{acyclic generic}};
      \node[box,below left=of plain] (skeletal) {\ul{skeletal generic}};
      \node[box,below left=of acyclic] (gaunt) {\ul{gaunt generic}};
      \draw[arrow] (weak) to (plain);
      \draw[arrow] (plain) to (acyclic);
      \draw[arrow] (plain) to (skeletal);
      \draw[arrow] (acyclic) to (gaunt);
      \draw[arrow] (skeletal) to (gaunt);
    \end{tikzpicture}
  \end{center}
  \vspace{1em}
  \caption{An analysis of the comparability of different notions of generic object, where the rightward direction represents (strictly) decreasing strength.}
  \label{fig:strength}
\end{figure}
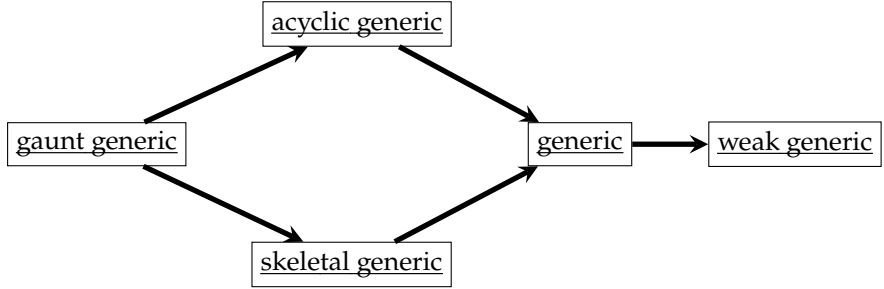

\subsection{Generic objects: ordinary, skeletal, and gaunt}

Let $\FibMor[p]{\ECat}{\BCat}$ be a fibered category. Our most basic definition below
corresponds to the \textbf{weak generic} objects of \citet{jacobs:1999}.

\begin{definition}
  A \ul{generic} object for $\FibMor[p]{\ECat}{\BCat}$ is defined to be
  an object $T\in\ECat$ such that for any $X\in \ECat$ there exists a cartesian
  morphism $\Mor{X}{T}$.
\end{definition}

\begin{definition}\label[definition]{def:generic-object:skeletal}
  A \ul{generic} object $T\in\ECat$ is called \ul{skeletal} when
  for any $X\in\ECat$, there exists a unique $\Mor{pX}{pT}$ such that there
  exists a cartesian map $\Mor{X}{T}$ lying over $\Mor{pX}{pT}$.
\end{definition}

\begin{definition}\label[definition]{def:generic-object:gaunt}
  A \ul{generic} object $T\in\ECat$ is called \ul{gaunt} when for
  any $X\in\ECat$, any two cartesian morphisms $\Mor{X}{T}$ are equal.
\end{definition}

\begin{scholium}\label[scholium]{schol:skeletal}
  Our \ul{skeletal} and \ul{gaunt} generic objects are the same as Jacobs'
  \textbf{generic} objects and \textbf{strong generic objects} respectively.
  Our terminology is inspired by a comparison between the properties of an
  internal category $\mathbb{C}$ and its externalization: in particular, the
  generic object over $\mathbb{C}_0$ is \ul{skeletal} when $\mathbb{C}$ is a
  skeletal category and it is \ul{gaunt} when $\mathbb{C}$ is a gaunt category.
  Unfortunately, \citet{phoa:1992} has used the word ``skeletal'' to describe
  what we call \ul{gaunt} generic objects; but this usage accords with \opcit's
  unconventional definition of a \emph{skeletal category}: usually a skeletal
  category is one in which any two isomorphic objects are equal, but Phoa
  defines it to be one in which the only isomorphisms are identity maps. Thus
  Phoa's skeletal categories would ordinarily be referred to as \emph{gaunt}
  categories.
\end{scholium}

To illustrate the comparison between (skeletal, gaunt) categories and
(\ul{skeletal}, \ul{gaunt}) generic objects respectively, we recall
\cref{lem:gaunt-skeletal-generic-object-comparison}, which is restated in the
new terminology as follows:

\begin{corollary}\label[lemma]{cor:gaunt-skeletal-generic-object-comparison:restated}
  Let $\mathbb{C}$ be a small category; then $\mathbb{C}$ is (respectively
  skeletal, gaunt) if and only if the \ul{generic} object
  $T=\brc{C}\Sub{C\in\mathbb{C}_0}$ for $\FAM{\mathbb{C}}$ is (respectively
  \ul{skeletal}, \ul{gaunt}).
\end{corollary}

\subsection{Weak generic objects and stack completions}\label{sec:stacks}

In this section, we are concerned with \emph{stacks} --- a 2-dimensional
variation on the concept of a sheaf.

\begin{definition}
  Suppose that a pretopos $\BCat$ is equipped with the structure of a site.
  Then a fibration $\FibMor{\ECat}{\BCat}$ is a stack when for any object $X\in
  \BCat$ and covering sieve $S\subseteq \BCat\downarrow X$, the canonical
  restriction functor
  $\Mor{\ECat_X}{\hom\Sub{\Kwd{Fib}\Sub{\BCat}}\prn{S,\ECat}}$ is an
  equivalence.\footnote{For more on stacks, see The Stacks
  Project~\citep{stacks-project} or the tutorial of \citet{vistoli:2004}.}
\end{definition}

Our examples will mainly concern stacks on a pretopos $\BCat$ equipped with its
regular topology.

\begin{definition}[{\citet[Definition~2.3.5]{frey:2013}}]
  A cartesian morphism $\Mor{X}{Y}$ in a fibered category
  $\FibMor[p]{\ECat}{\BCat}$ is called \emph{cover-cartesian} when
  $\Mor|->>|{pX}{pY}$ is a regular epimorphism.
\end{definition}

\begin{definition}\label[definition]{def:weak-generic}
  A \ul{weak generic} object for $\FibMor[p]{\ECat}{\BCat}$ is defined
  to be an object $T\in\ECat$ such that for all $X\in\ECat$ there exists a
  span of cartesian maps $X\leftarrow\tilde{X}\to T$ where
  $\Mor{\tilde{X}}{X}$ is cover-cartesian.
\end{definition}

\begin{scholium}
  Our definition and terminology agrees with that of
  \citet[(5.2)]{streicher:2005} and
  \citet[Definition~6.3.7]{battenfeld:2008:thesis}, differing only in that both
  of the cited works specialize the definition for B\'enabou-definable full
  subfibrations of the fundamental fibration of a topos, \ie full subfibrations
  that are stacks. Our definition is also identical to the \emph{representability
  axiom} (S2) for classes of small maps in algebraic set
  theory~\citep[Definition~1.1]{joyal-moerdijk:1995}.
\end{scholium}

\begin{remark}
  The notion of \ul{weak generic} object defined above should be thought of as an
  \emph{internal} version of the property of being a \ul{generic} object ---
  internal to a fibration that is a stack for the regular topology. Indeed, our
  explicit definition ought to be the translation of ordinary genericity
  through a generalization of Shulman's stack semantics~\citep{shulman:2010}
  that is stated for stacks other than the fundamental fibration
  $\FibMor{\FFib{\BCat} = \BCat^\to}{\BCat}$. The same (informal) translation is
  used by \citet{hyland-robinson-rosolini:1990} to correctly define \emph{weak
  equivalences} and \emph{weak completeness} for categories fibered over a
  regular category.
\end{remark}

\ul{Weak generic} objects in the sense of \cref{def:weak-generic} arise
very naturally.

\nocite{bunge:1979}
\begin{example}\label[example]{ex:stack-completion}
  Let $\Mor[\pi]{E}{U}$ be a morphism in a topos
  $\BCat$; then the class of maps arising as pullbacks of $\Mor[\pi]{E}{U}$
  determines a full subfibration $\brk{\pi}\subseteq\FFib{\BCat}$ of the
  fundamental fibration $\FibMor{\FFib{\BCat}=\BCat^\to}{\BCat}$ for which
  $\Mor[\pi]{E}{U}$ is a \ul{generic} object.
  The \emph{stack completion} of $\brk{\pi}$ is a weakly (but not necessarily
  strongly) equivalent full subfibration $\brc{\pi}\subseteq \FFib{\BCat}$, and
  it can be computed like so: an object of $\brc{\pi}\Sub{I}$ is a morphism
  $\Mor{X}{I}$ such that there exists a regular epimorphism
  $\Mor|->>|{\tilde{I}}{I}$ such that the pullback
  $\Mor{X\times\Sub{I}\tilde{I}}{\tilde{I}}$ lies in
  $\brk{\pi}\Sub{\tilde{I}}$. In other words, we have the pullback squares in
  the following configuration:
  \[
    \begin{tikzpicture}[diagram]
      \SpliceDiagramSquare<l/>{
        north/style = {<<-},
        south/style = {<<-},
        ne/style = {pullback, ne pullback},
        nw = X,
        sw = I,
        ne = X\times\Sub{I}\tilde{I},
        se = \tilde{I},
        height = 1.5cm,
      }
      \SpliceDiagramSquare<r/>{
        glue = west, glue target = l/,
        ne = E,
        se = U,
        height = 1.5cm,
      }
    \end{tikzpicture}
  \]

  Unless $\brk{\pi}$ was already a stack, it is not necessarily the case that
  $\Mor{E}{U}$ is a \ul{generic} object for the stack completion
  $\brc{\pi}$. But \citet{streicher:2005} points out that $\Mor{E}{U}$ is
  nevertheless a \ul{weak generic} object for $\brc{\pi}$, essentially
  by definition.
\end{example}

Scenarios of the form described in \cref{ex:stack-completion} are easy to come
by; a canonical example is furnished by the modest objects of a realizability
topos as described by \citet{hyland-robinson-rosolini:1990}, with critical
implications for the denotational semantics of polymorphism.

\begin{example}[\citet{hyland-robinson-rosolini:1990}]\label[example]{ex:eff-pers}
  Let $\Kwd{Eff}$ be the \emph{effective topos}~\citep{hyland:1982}, and let
  $\Con{N}$ be its object of realizers, \ie the partitioned assembly given by
  $\mathbb{N}$ such that $n\Vdash m \in \Con{N} \Leftrightarrow m = n$. Write
  $\Con{P}\subseteq \Omega\Sup{\Con{N}\times\Con{N}}$ for the the object of
  $\lnot\lnot$-closed partial equivalence relations on $\Con{N}$ and
  $\Con{P}'\in \Kwd{Eff}\downarrow \Con{P}$ for the \emph{generic
  $\lnot\lnot$-closed subquotient} of $\Con{N}$; internally, this is the
  subquotient $R:\Con{P}\vdash \Compr{x:\Con{N}}{x\mathrel{R}x}/R$.

  The family $\Mor[\pi]{\Con{P}'}{\Con{P}}$ then induces a full subfibration
  $\brk{\pi}\subseteq\FFib{\Kwd{Eff}}$ spanned by morphisms that arise by
  pullback from $\pi$; an element $\brk{\pi}_I$ is an object at stage $I$ that
  is the subquotient of $\Con{N}$ by some partial equivalence defined at stage
  $I$. The fibration $\FibMor{\brk{\pi}}{\Kwd{Eff}}$ is small with
  \ul{generic} object $\pi$, but it is \emph{not} complete.  Although
  for every $\Kwd{Eff}$-indexed diagram $\Mor{\mathbb{C}}{\brk{\pi}}$ there
  ``exists'' in the internal sense a limit, we cannot globally choose this
  limit. This situation is referred to by \citet{hyland-robinson-rosolini:1990}
  as \emph{weak completeness} as opposed to \emph{strong completeness}.

  In contrast, we may consider the stack completion $\brc{\pi}$ of $\brk{\pi}$.
  An element of $\brc{\pi}_I$ is given by an object $E$ at stage $I$ such that
  there ``exists'' (in the internal sense) a partial equivalence relation that
  it is the quotient of --- externally, this means that there is a regular
  epimorphism $\Mor{\tilde{I}}{I}$ such that $\tilde{I}^*E$ is the subquotient
  of $\Con{N}$ by some partial equivalence relation defined at stage
  $\tilde{I}$. The stack $\brc{\pi}$ is weakly but not strongly equivalent to
  $\brk{\pi}$; on the other hand, $\brc{\pi}$ is complete in the strong sense.
  Finally, we observe that $\Mor[\pi]{\Con{P}'}{\Con{P}}$ is a \ul{weak
  generic} object for the stack completion $\brc{\pi}$.

  If we pull back $\brk{\pi}$ to along the inclusion
  $\EmbMor[i]{\Kwd{Asm}}{\Kwd{Eff}}$ of assemblies / $\lnot\lnot$-separated
  objects in $\Kwd{Eff}$, then we have a \emph{complete} fibered category
  $i^*\brk{\pi}$ over $\Kwd{Asm}$ which turns out to be (strongly) equivalent
  to the familiar fibration $\FibMor{\Kwd{UFam}\prn{\Kwd{PER}}}{\Kwd{Asm}}$ of
  uniform families of PERs over assemblies. In fact, $i^*\brk{\pi}$ is strongly
  equivalent to $i^*\brc{\pi}$ --- in other words, over assemblies
  there is no difference between an object that is locally isomorphic to a
  subquotient of $\Con{N}$ and an actual subquotient of $\Con{N}$.
\end{example}

There is another way to think of the situation described at the end of
\cref{ex:eff-pers}, using the observations of \citet{streicher:2005} on the
relationship between B\'enabou's notion of \emph{definable class} and
stackhood: the class of families of subquotients of $\Con{N}$ is not definable
in $\Kwd{Eff}$, but it is definable in $\Kwd{Asm}$.

\subsection{A new class: \texorpdfstring{\ul{acyclic}}{acyclic} generic objects}\label{sec:acyclic}

Inspired by the crucial \emph{realignment} property of type theoretic
universes~\citep{gratzer-shulman-sterling:2022:universes} that was discussed in \cref{sec:consequences:tt}(1), we define a new kind
of generic object for an arbitrary fibration that restricts in the case of a
full subfibration to a universe satisfying realignment. We refer to
\citet{gratzer-shulman-sterling:2022:universes} for an explanation of the
applications of realignment.
In this section, let $\FibMor[p]{\ECat}{\BCat}$ be a fibration and let
$\mathcal{M}$ be a class of monomorphisms in $\BCat$.

\begin{definition}
  A \ul{generic} object $T$ for $\FibMor[p]{\ECat}{\BCat}$ is called
  \emph{$\mathcal{M}$-\ul{acyclic}} when for any span of cartesian maps
  in $\ECat$ as depicted below in which $\Mor|>->|{pU}{pX}$ lies in $\mathcal{M}$,
  \[
    \begin{tikzpicture}[diagram]
      \node (U) {$U$};
      \node (T) [right = of U] {$T$};
      \node (X) [below = of U] {$X$};
      \draw[->] (U) to node [above] {cart.} (T);
      \draw[->] (U) to node [sloped,below] {cart.} (X);
      \node[color=gray] (pU) [right = 2cm of T] {$pU$};
      \node[color=gray] (pX) [below = of pU] {$pX$};
      \draw[color=gray,>->] (pU) to node [right] {$\in \mathcal{M}$} (pX);
    \end{tikzpicture}
  \]
  there exists a cartesian map $\Mor{X}{T}$ making the following diagram commute in $\ECat$:
  \[
    \begin{tikzpicture}[diagram]
      \node (U) {$U$};
      \node (T) [right = of U] {$T$};
      \node (X) [below = of U] {$X$};
      \draw[->] (U) to node [above] {cart.} (T);
      \draw[->] (U) to node [sloped,below] {cart.} (X);
      \draw[exists,->] (X) to node [sloped,below] {$\exists$ cart.} (T);
    \end{tikzpicture}
  \]
\end{definition}

\begin{convention}
  When $\mathcal{M}$ is understood to be the class of \emph{all} monomorphisms in
  $\BCat$, we will simply speak of \ul{acyclic generic} objects.
\end{convention}

All of the examples of \ul{acyclic generic} objects that we are aware of so
far have arisen in the context of the full subfibration spanned by a
\emph{universe} in the sense of \citet{streicher:2005}, where \ul{acyclicity}
reduces to the \emph{realignment} property studied in detail by
\citet{gratzer-shulman-sterling:2022:universes}.

\begin{remark}
  It is reasonable to ask whether there is a ``weak'' version of
  \ul{acyclicity} that pertains to stack completions in the same way that
  \ul{weak generic} objects relate to \ul{generic} objects, \eg by asking for a
  cover-cartesian morphism $\Mor|->>|{\tilde{X}}{X}$ and an extension of
  $U\times_X\tilde{X}\to{U}\to T$ along $\Mor{U\times_X \tilde{X}}{\tilde{X}}$. It
  remains somewhat unclear to this author whether this notion is useful, so we
  do not include it in our proposal.
\end{remark}

\begin{scholium}
  In the context of universes, the \emph{realignment} or \ul{acyclicity}
  condition has been referred to
  as ``Axiom ($2'$)'' by \citet{shulman:2015:elegant},
  as ``strictification'' by \citet{orton-pitts:2016},
  as ``stratification'' by \citet{stenzel:2019:thesis},
  as ``alignment'' by \citet{awodey:2021:qms},
  and
  as ``strict gluing'' by \citet{sterling-angiuli:2021}.
  See also \citet{riehl:2022:cirm} for further discussion.
\end{scholium}

\paragraph*{Origins of the terminology}

The origin of the term ``acyclicity'' is explained by \citet{shulman:2019} and
\citet{riehl:2022:cirm} in essentially the following way. Let $\ECat\subseteq
\FFib{\BCat}$ be a full subfibration equipped with a generic object $T$; we
will write $\Core{\ECat}$ for the \emph{groupoid core} of $\ECat$, and for any
$I\in\BCat$ we will write $\FibMor{\Yo{\BCat}{I}}{\BCat}$ for the discrete
fibration whose fiber at $J\in \BCat$ is the set of morphisms $\Mor{J}{I}$.
Then we have a canonical morphism of fibered categories
$\Mor{\Yo{\BCat}{pT}}{\Core{\ECat}}$ corresponding under the fibered Yoneda
lemma~\citep[\S3]{streicher:2021:fib} to $T$ itself. Realignment for $\ECat$ is
then the property that $\Mor{\Yo{\BCat}{pT}}{\Core{\ECat}}$ has the following
(2-categorical) extension property with respect to any monomorphism
$\Mor|>->|[m]{J}{I}\in\mathcal{M}$, in which we write
$\Mor[\floors{X}]{\Yo{pX}}{\Core{\ECat}}$ for the morphism that we identify with
$X\in \ECat$ under the fibered Yoneda lemma:
\[
  \begin{tikzpicture}[diagram]
    \SpliceDiagramSquare{
      west/style = >->,
      nw = \Yo{\BCat}{J},
      sw = \Yo{\BCat}{I},
      ne = \Yo{\BCat}{pT},
      se = \Core{\ECat},
      east = \floors{T},
      west = \Yo{\BCat}{m},
      south = \floors{B},
      north = \Yo{\BCat}{\alpha},
    }
    \draw[->,exists] (sw) to node [desc] {$\exists \Yo{\BCat}{\beta}$} (ne);
  \end{tikzpicture}
\]

More formally, for any square as above commuting up to an isomorphism $\phi :
\floors{B}\circ \Yo{\BCat}{m} \cong \floors{T}\circ \Yo{\BCat}{\alpha}$, the
extension property states that there exists a (not necessarily unique) morphism
$\Mor[\beta]{I}{pT}$ and isomorphisms $\phi_0 : \Yo{\BCat}{\beta}\circ
\Yo{\BCat{m}} \cong \Yo{\BCat}{\alpha}$ and $\phi_1 : \floors{T}\circ
\Yo{\BCat}{\beta}\cong \floors{B}$ such that $\phi$ factors as the pasting of
$\phi_0$ and $\phi_1$.  Moreover, as the boundary of $\phi_0$ is discrete, it
must be an identity and thus $\Mor[\beta]{I}{pT}$ extends $\Mor[\alpha]{J}{pT}$
along $\Mor|>->|{J}{I}$ in the (strict) 1-categorical sense.
Under the fibered Yoneda lemma, the isomorphism $\phi : \floors{B}\circ \Yo{\BCat}{m}
\cong \floors{T}\circ \Yo{\BCat}{\alpha}$ corresponds to an isomorphism
$\ceils{\phi} : m^*B\cong \alpha^*T$ and the isomorphism $\phi_1 :
\floors{T}\circ \Yo{\BCat}{\beta}\cong\floors{B}$ corresponds to an isomorphism
$\ceils{\phi_1} : \beta^*T\cong B$ that extends $\ceils{\phi}$ along
$\Mor|>->|[m]{J}{I}$.

In a model categorical setting where $\mathcal{M}$ is understood to be the
class of cofibrations, the extension property above expresses that
$\Mor{\Yo{\BCat}{pT}}{\Core{\ECat}}$ is an \emph{acyclic fibration}, whence the terminology.

\subsubsection{Incomparability of \texorpdfstring{\ul{acyclic}}{acyclic} and \texorpdfstring{\ul{skeletal generic}}{skeletal generic} objects}

It is not the case that a \ul{skeletal} generic object need be
\ul{acyclic}. We may consider the following counterexample, which was kindly
suggested by one of the anonymous referees of this paper; in what follows,
assume that $\BCat$ is a category with finite limits and that $G$ is a group
object in $\BCat$. As in \cref{ex:delooping}, we may regard $G$ as an internal
groupoid $\Deloop{G}$ in $\BCat$ with a single object $\prn{\Deloop{G}}_0 =
1_{\BCat}$, such that the hom object $\prn{\Deloop{G}}_1$ is exactly $G$.  The
externalization $\brk{\Deloop{G}}$ can be seen to have the same objects as
$\BCat$; a morphism $\Mor{I}{J}$ in $\brk{\Deloop{G}}$ is given by a pair of a
morphism $\Mor[f]{I}{J}$ together with a generalized element $\Mor[u]{I}{G}$.
As in any fibered groupoid, every morphism of $\brk{\Deloop{G}}$ is cartesian.
Observe that the object $T = 1\Sub{\BCat}$ is a \ul{skeletal generic} object in
$\brk{\Deloop{G}}$; below, we show that $T$ need not be \ul{acyclic}.

\begin{lemma}
  The \ul{skeletal generic} object $T\in \brk{\Deloop{G}}$ is
  \ul{acyclic} if and only if $G\in\BCat$ has the right lifting property with respect
  to any element of $\mathcal{M}$.
\end{lemma}

\begin{proof}
  Suppose that $T$ is an \ul{acyclic generic} object, and fix a monomorphism
  $\Mor|>->|[m]{J}{I}\in \mathcal{M}$ together with a generalized element
  $\Mor[g]{J}{G}$ in $\BCat$. We may exhibit a span of cartesian maps in
  $\brk{\Deloop{G}}$ as follows, which has a lift by \ul{acyclic}ity:
  \[
    \begin{tikzpicture}[diagram]
      \node (U) {$J$};
      \node (T) [right = of U] {$T$};
      \node (X) [below = of U] {$I$};
      \draw[->] (U) to node [above] {$\prn{!\Sub{J}, g}$} (T);
      \draw[->] (U) to node [left] {$\prn{m,\epsilon}$} (X);
      \draw[exists,->] (X) to node [sloped,below] {$\prn{h,\hat{g}}$} (T);
    \end{tikzpicture}
  \]

  Above we have exactly a morphism $\Mor[\hat{g}]{I}{G}$
  extending $\Mor[g]{J}{G}$ along $\Mor|>->|[m]{J}{I}$.
  The other direction is analogous.
\end{proof}

It is clear that an \ul{acyclic generic} object need not be \ul{skeletal}; the
above discussion confirms that \ul{acyclic} and \ul{skeletal generic} objects
are in fact incomparable.

\subsection{Splittings and generic objects}

We have shown in \cref{sec:weak-and-split} that the correct generalization to
non-split fibrations of a \textbf{split generic} object in the sense of
\citet{jacobs:1999} is what we have proposed to call a \ul{generic} object, \ie
the \textbf{weak generic} object of \opcit. Thus we conclude that the correct
relationship can be established between \emph{split $\blacksquare$-fibrations}
and \emph{$\blacksquare$-fibrations} when re-expressed using our definitions, where $\blacksquare$ ranges over the different kinds of fibered structures (\eg $\lambda2$-fibration, polymorphic fibration, \etc).
For example, the following definition expresses the correct relationship
between $\lambda2$-fibrations and split $\lambda2$-fibrations:

\begin{definition}
  A \ul{$\lambda2$-fibration} is a fibration with a \ul{generic} object $T$, fibered
  finite products, and simple $pT$-products and coproducts. A
  \ul{$\lambda2$-fibration} will be called \ul{split} if all its structure is
  split.
\end{definition}

\subsection{In univalent mathematics}

The theory of fibered (1-)categories carries over \emph{mutatis mutandis} to
the univalent foundations of mathematics, as shown by
\citet{ahrens-lumsdaine:2019}. The pertinent notions of generic object are,
however, a bit different in a univalent setting. In particular, \ul{gaunt}
generic objects play a more fundamental role in univalent mathematics than in
non-univalent mathematics, because the uniqueness of the classifying squares
becomes a statement of contractibility rather than strict uniqueness. In
non-univalent mathematics, uncontrived instances of \ul{gaunt} generic objects
tend to follow the pattern of subobject classifiers in representing a fibered
poset; in univalent mathematics, we also have \emph{object} classifiers
(univalent universes) which are \ul{gaunt} generic objects for groupoid cores
of the full subfibrations they induce.

Just as in non-univalent mathematics, \ul{skeletal generic} objects are not
expected to play a significant structural role although they may appear when
considering the deloopings of groups, as in our 1-categorical examples. The
most interesting case is that of \ul{acyclic generic} objects: the
\ul{acyclicity} property is the correct way to present in non-univalent
mathematics a generic object that is \ul{gaunt} \emph{in the homotopical
sense}, as Shulman has recently argued~\citep{shulman:2022:grothendieck:talk}.
Thus in a homotopical/univalent setting, one expects to work directly with
\ul{gaunt generic} objects rather than \ul{acyclic} ones.

\section{Concluding remarks}

For a number of years, the disorder in the variants of generic object has led
to a proliferation of subtle differences in terminology between different
papers applying fibered categories to categorical logic, type theory, and the
denotational semantics of polymorphic types. Based on the kinds of generic
object that occur most naturally or have the most utility, we have proposed a
unified terminological scheme for generic objects that we believe will meet the
needs of scientists working in these areas.

\subsection*{Acknowledgments}
I thank Lars Birkedal, Daniel Gratzer, Patricia Johann, and Peter LeFanu
Lumsdaine for helpful conversations on the topic of this paper. I am
particularly grateful to the anonymous referees whose suggestions and
corrections have substantially improved this paper.
This work was supported by a Villum Investigator grant (no.~25804), Center for
Basic Research in Program Verification (CPV), from the VILLUM Foundation. This
work was co-funded by the European Union. Views and opinions expressed are
however those of the author only and do not necessarily reflect those of the
European Union or the European Commission. Neither the European Union nor the
granting authority can be held responsible for them.

\bibliographystyle{plainnat}
\bibliography{references/refs-bibtex}

\end{document}